\documentclass[12pt]{article}

\usepackage{natbib}
\usepackage{amsmath,amssymb}
\usepackage[hyperref,amsmath,amsthm,thmmarks]{ntheorem}

\newtheorem{theorem}{Theorem}

\newtheorem{proposition}[theorem]{Proposition}
\newtheorem{corollary}[theorem]{Corollary}
\newtheorem{claim}[theorem]{Claim}

\usepackage{hyperref}

\usepackage{graphicx,comment}

\newcommand{\FIXME}[1]{}
\newcommand{\GE}{{\cal G}}
\newcommand{\out}{32 \log(n)}
\newcommand{\FS}{{\cal F_S}}
\newcommand{\WS}{{\cal W_S}}
\newcommand{\nneigh}{2}
\newcommand{\loss}{\frac{512\log(n)}{g(1) n}}
\newcommand{\thresh}{\tau}
\newcommand{\kp}{k}

\title{An Approximate ``Law of One Price'' in Random Assignment Games}

\author{Avinatan Hassidim\thanks{Bar-Ilan University, Department of Computer Science} \and Assaf Romm\thanks{Harvard University, Department of Economics and HBS}}

\begin{document}

\maketitle

\begin{abstract}
Assignment games represent a tractable yet versatile model of two-sided markets with transfers. We study the likely properties of the core of randomly generated assignment games. If the joint productivities of every firm and worker are i.i.d bounded random variables, then with high probability all workers are paid roughly equal wages, and all firms make similar profits. This implies that core allocations vary significantly in balanced markets, but that there is core convergence in even slightly unbalanced markets. For the benchmark case of uniform distribution, we provide a tight bound for the workers' share of the surplus under the firm-optimal core allocation. We present simulation results suggesting that the phenomena analyzed appear even in medium-sized markets. Finally, we briefly discuss the effects of unbounded distributions and the ways in which they may affect wage dispersion.
\end{abstract}

\section{Introduction} \label{sec:introduction}

\FIXME{Explain the expression whp somewhere}

The ``law of one price'' asserts that homogeneous goods must sell for the same price across locations and vendors. In labor markets, it implies that workers who are equally skilled should earn the same salary. This paper makes the formal claim that even in the presence of some heterogeneity across firms and workers, an approximate law remains valid.\footnote{We choose to focus on labor markets as our main example. However, our results are also applicable to markets with heterogeneous commodities and unit demand buyers.} While the law of one price does not sit well with empirical evidence, our model provides a useful benchmark for testing hypotheses on possible sources of wage inequality. Furthermore, the analysis elucidates how surplus related to idiosyncratic compatibility is divided among market participants.

Our basic building block is the assignment game model of \cite{ss1971}. Each firm from a finite set of firms $F$ is looking to hire exactly one worker from a finite set of workers $W$, in exchange for a negotiable salary. Each firm has a (possibly different) value for hiring each of the workers, and each worker has a (possibly different) reservation value for working for each of the firms. Transfers are permitted between any two parties (not just a transfer from a firm to a worker employed by it), and utility is assumed to be linear in money. Note that since transfers are freely allowed, we can describe the net productivity of each firm-worker pair by a single number.

We bolster the model by assuming that each element in the matrix of productivities is independently and identically distributed on a bounded interval. This assumption is similar in spirit to the one taken in many auction theory papers, in which bidders' valuations are assumed to be heterogeneous and determined according to some random distribution. However, unlike most of the literature on auction theory, we do not wish to study the effects of the random generation on agents' beliefs and equilibrium behavior. Instead we take a different approach and characterize the likely outcomes in a typical complete information matching market created in that manner.

The main result of this paper is that under certain regularity conditions and with high probability, in any core allocation workers are being paid approximately identical salaries, and firms receive approximately identical profits. Because of inherent heterogeneity in firms' and workers' preferences, the law of one price only holds approximately, with some workers being paid more than their peers. Nevertheless, the differences become negligible as the market becomes large, and there is no wage dispersion in a reasonably-sized random assignment game. Formally, we show that with high probability the core of an assignment game generated by i.i.d draws from a bounded random distribution with continuous and positive density is narrow, in the sense that the difference between the payoffs of agents on the same side of the market are $O\left(\frac{\log^2(n)}{n}\right)$.

The approximate law of one price allows a deeper understanding of surplus division between the two sides of the market. It was shown already by \cite{ss1971} that there are core allocations which are optimal for the firms, and core allocations that are optimal for the workers. In a market with different number of firms and workers it is immediately implied that agents on the long side will get a vanishing share of the total surplus. However, in a balanced market the choice of core allocation can have compelling consequences for all agents in the market. This extends the basic economic intuition on competition in a market for homogeneous goods.

It is worthwhile to mention that the model can also be used to describe  auctions of heterogeneous items with unit-demand bidders \cite{dgs1986}. With this interpretation in mind, core allocations are equivalent to Walrasian equilibira. Our results therefore provide insight into revenue acquired by sellers under different market conditions.

From a slightly more broad point of view, this paper belongs to the theoretical literature on matching in two-sided markets. This literature gained prominence in the 1960's and early 1970's following the publication of the seminal papers by \cite{gs1962} and \cite{ss1971}, yet research remained mostly divided (with some notable exceptions) into two parallel strands: with and without transferable utility (i.e.\ money). During the 1980's, as it became clear that real-life centralized clearing houses can be immensely improved using intuitions gained in the study of marriage markets, the transferable utility strand of the literature became slightly neglected compared to its glorified non-transferable utility half-sibling. It is our belief that assignment games provide an excellent way to model decentralized markets, and that both strands of the matching theory literature can benefit from the continuous cross-fertilization.

The rest of the paper is organized as follows. \autoref{sec:related_work} reviews the literature related to our paper. \autoref{sec:model} introduces the model and the formal notation. \autoref{sec:approx_lop} contains the statement and proof of the main result, and \autoref{sec:applications} applies the result to analyze the likely properties of balanced and unbalanced markets. \autoref{sec:simulations} provides some simulation results and a discussion on unbounded distribution, and \autoref{sec:conclusion} concludes.

\section{Related Work} \label{sec:related_work}

Assignment games were introduced by \cite{shapley1955} and thoroughly analyzed by \cite{ss1971}. A (non-comprehensive) list of further work on assignment games includes the study of strategic incentives \cite{dg1985}, entry \cite{mo1988}, convergence via decentralized processes \cite{cfy2012,kp2012,npy2013}, elongation of the core \cite{quint1987} and its dimensions \cite{nr2008}, and median stable matchings \cite{sy2011}.

Analysis of random instances of the linear sum assignment problem was first carried by \cite{walkup1979}, and was subsequently improved by \cite{aldous2001,cs1999,karp1987}. For a survey on the topic and other related literature see \cite{kp2009assignment}.

While to the best of our knowledge we are the first to introduce random values into the assignment game framework, some of the consequences we present in \autoref{sec:applications} resemble work done on random marriage markets. Pioneered by \cite{wilson1972} and \cite{knuth1976}, and extensively developed by \cite{pittel1989,pittel1992}, the analysis of random preference marriage markets suggests that in a situation in which the number of men is equal to the number of women, with high probability the proposing side's (in a deferred acceptance algorithm) mean rank of partners behaves like $\log(n)$, where as the other side's mean rank of partners behaves like $\frac{n}{\log(n)}$. Recently, \cite{akl2013} have shown that in unbalanced random marriage markets with high probability under any stable matching the short side's mean rank of partners behaves like $\log(n)$, whereas the long side's mean rank of partners behaves like $\frac{n}{\log(n)}$. \cite{cgs2014} and \cite{cs2014} utilize those results to study aspects of strategic behavior in marriage markets with incomplete information. The papers most closely related to the current paper are \cite{lee2014} and \cite{ly2014} that assume that preferences are derived from underlying cardinal utilities and study the issues of core convergence and efficiency, respectively. Finally, a different type of results on random marriage markets (and extensions thereof) are related to showing core convergence when preferences are bounded in length \cite{abh2011,im2005,kp2009,kpr2013}.\footnote{Similar ideas were also applied by \cite{manea2009}, \cite{ck2010} and \cite{km2010}.}

\section{Model and Notation} \label{sec:model}

Consider a sequence of markets $\{M^n\}_{n=1}^\infty$, such that each market can be described as $M^n = \left(F^n, W^n, \alpha^n\right)$, where $F^n$ is a set of firms of size $n$, $W^n$ is a set of workers of size $n + k(n)$, with $k(n) \geq 0$, and $\alpha^n$ being an $|F^n| \times |W^n|$ real matrix representing the value of pairs of firms and workers. We assume throughout each element of $\alpha^n$ is distributed i.i.d according to the cdf $G$ which is bounded on the interval $[0,1]$ (meaning $G(1) = 1$) and has a pdf $g$ which is continuous and strictly positive\footnote{In fact for our our results to hold we just need that the pdf is continuous near its supremum, and that it doesn't depend on $n$. In addition, the results also hold if there are atoms. Specifically, if there is an atom at the supremum of the distribution then whp the surplus will be $n$ times the supremum, there will be no wage dispersion, and in a firm-optimal allocation all workers will get zero. Atoms at other places do not matter as $n$ goes to infinity.}.

In market $M^n$, the value of a coalition of firms and workers $S$ is given by
\[ \mathrm{v}(S) = \max \left[ \alpha^n_{i_1 j_1} + \alpha^n_{i_2 j_2} + \dots + \alpha^n_{i_l j_l} \right], \]
where the maximum is taken over all arrangements of $2l$ distinct agents, $i_1,\dots,i_l \in S \cap F$, $j_1,\dots,j_l \in S \cap W$, and $l \leq \min \left\{ \left|S \cap F\right|, \left|S \cap W \right| \right\}$. An \textbf{allocation} is denoted by $(\mu,u,v)$ with $\mu$ being a matching of firms to workers and vice-versa, and $u$ and $v$ being payoff vectors for the firms and workers respectively. We will refer to $u$ as firms' ``profits'', and to $v$ as workers' ``salaries''. Formally, $\mu \colon F^n \cup W^n \rightarrow F^n \cup W^n \cup \{\emptyset\}$, and satisfies
\begin{enumerate}
\item $\forall f \in F^n: \mu(f) \in W^n \cup \{\emptyset\}$,
\item $\forall w \in W^n: \mu(w) \in F^n \cup \{\emptyset\}$, and
\item $\forall f \in F^n, w \in W^n: \mu(f) = w \iff \mu(w) = f$.
\end{enumerate}
An allocation is a \textbf{core allocation} if no coalition can deviate and split the resulting value between its members such that each member of the coalition becomes strictly better off. We denote the set of core allocations of $M^n$ by $C\left(M^n\right)$. \cite{ss1971} have shown that the core is a non-empty compact and convex set, and that it is elongated in the sense that there is a firm-optimal core allocation in which salaries are the lowest, and a worker-optimal core allocation in which salaries are the highest.

\section{An approximate law of one price} \label{sec:approx_lop}

Our main theorem states that the core is not only elongated, but also narrow, in the sense that with high probability (i.e.\ with probability that approaches 1) firms' profits and workers' salaries do not vary a lot.

\begin{theorem} \label{th:approx_lop}
There exists $c \in \mathbb{R}_+$ such that with probability at least $1-O\left(\frac{1}{n}\right)$, for any $(\mu,u,v) \in C\left(M^n\right)$ we have
\begin{enumerate}
\item $\forall i,j \in \left\{1,\dots\left|F^n\right|\right\}: \left| u_i - u_j \right| \leq \frac{c \log^2 n}{n}$, and
\item $\forall i,j \in \left\{1,\dots\left|W^n\right|\right\}: \left| v_i - v_j \right| \leq \frac{c \log^2 n}{n}$.
\end{enumerate}

\end{theorem}

Let $\GE$ be a directed bipartite graph where firms are vertices of one side, and workers are vertices on the other side. Each firm $f$ points to the $32 \frac{n + \kp(n)}{n} \log(n)$ workers it values the most, equivalently there is an edge $(f,w)$ if and only if
\[\left|\left\{\alpha_{f,j} : \alpha_{f,j} \ge \alpha_{f,w}\right\}\right| \le 32 \frac{n + \kp(n)}{n} \log(n).\]
\noindent Note that since there are no ties in $\alpha$ the verbal definition coincides with the mathematical one.

Each worker that is assigned points to the firm which hired her.

\begin{claim} \label{expander}
With probability at least $1 - 1/n$, for any set of firms $\FS$ with $|\FS| < n/10$, we have that
\[\left|N(\FS)\right| \ge 2\frac{n + \kp(n)}{n}  |\FS|,\]
where $N(S)$ is the set of neighbors of a set of vertices $S$. In addition, for any set of firms $|\FS|$ of size $n/10$
\[\left|N(\FS)\right| \ge 0.99(n + \kp(n)).\]
\end{claim}

\begin{proof}
Consider a set of firms $\FS$ of size $m < n/10$, and a set of workers $\WS$ of size $\eta < 2\frac{n + \kp(n)}{n} m$. Since the edges coming out of $\FS$ point to independent random workers, the probability that all edges point towards workers in $\WS$ is at most
\[ \left(\frac{\eta}{n + \kp(n)}\right)^{\frac{32(n + \kp(n))m \log n}{n}}.\]
Taking a union bound over all sets of firms of size $m$, and all sets of workers of size $\eta$, the probability that there is a set of firms of size $m$ which points to $\eta < 2\frac{n + \kp(n)}{n} m$ workers is upper bounded by
 \begin{eqnarray*} \binom{n}{m}\binom{n+ \kp(n)}{\eta}\left(\frac{\eta}{n + \kp(n)}\right)^{\frac{32(n + \kp(n))m \log n}{n}} & \le & \\
 \binom{2n}{2\frac{n + \kp(n)}{n} m}^2 \left(\frac{ 2\frac{n + \kp(n)}{n} m }{n + \kp(n)}\right)^{\frac{32(n + \kp(n))m \log n}{n}} & = & \\
 \binom{2n}{2\frac{n + \kp(n)}{n} m}^2 \left(\frac{ 2 m }{n}\right)^{\frac{32(n + \kp(n))m \log n}{n}} & \le & \\
 \binom{2n}{4 m}^2 \left(\frac{ 2 m }{n}\right)^{32 m \log(n)} & \le &\\
 \left(\frac{2n e}{4m}\right)^{8m} \left(\frac{ 2 m }{n}\right)^{32 m \log(n)} & \le & 1/n^2,
 \end{eqnarray*}
where the second inequality uses $0 \le \kp(n) \le n$, the third uses the fact that $m < n/10$ and binomial coefficients increase until $\binom{2n}{n}$, and the fourth is an approximation of the binomial coefficient.

Using $m < n/10$, we union bound over $m$, and get a failure probability of at most $1 / 10n$.

Similarly, let $\FS$ be a set of workers with $|\FS| = n/10$. The probability that all the workers they point to lie in a set of size at most $0.99(n + \kp(n)$ is

 \begin{eqnarray*}
   \binom{n}{n/10}\binom{n + \kp(n)}{0.99(\kp(n)+ n)} \left(\frac{0.99(n + \kp(n))}{n + \kp(n)}\right)^{3.2n \log(n)} & = & \\
      \binom{n}{n/10}\binom{n + \kp(n)}{0.01(\kp(n)+ n)} \left(\frac{0.99(n + \kp(n))}{n + \kp(n)}\right)^{3.2 n \log(n)} & \le & \\
   \left(\frac{ne}{n/10}\right)^{n/10}\left(\frac{(n + \kp(n))e}{(n + \kp(n))/100}\right)^{(n + \kp(n))/100} \left(\frac{0.99(n + \kp(n))}{n + \kp(n)}\right)^{3.2 n \log(n)} & \le & \\
   (10e)^{n/10}(100e)^{0.02n}\left(\frac{1}{100}\right)^{3 n \log(n)} \le 1/n^2,
 \end{eqnarray*}
 which finishes the proof.
\end{proof}

If $\kp(n) > n$, we can prove the following stronger claim, via a similar technique.

\begin{claim} \label{expander_large_kappa}
With probability at least $1 - 1/n$, for any set of firms $\FS$, we have that
\[\left|N(\FS)\right| \ge 2\frac{n + \kp(n)}{n}  |\FS|,\]
where $N(S)$ is the set of neighbors of a set of vertices $S$.
\end{claim}

For $\kp(n) \le n$, we also show a similar claim for workers, showing that for any set of workers, the set of firms which point to them is large.

\begin{claim} \label{point_to_workers}
With probability at least $1 - 1/n$, for any set of workers $\WS$ with $|\WS| < n/10$, we have that
\[\left|\{f \ : \ N(f) \cap \WS \neq \emptyset\}\right| \ge 2 |\WS|.\]
\end{claim}

\begin{proof}
Consider a set of workers $\WS$ of size $m < n/10$, and a set of firms $\FS$ with $|\FS| = \eta < 2m$. The probability that for every $f \not \in \FS$ we have $N(f) \cap \WS = \emptyset$ is
\[\left(\frac{n + \kp(n) - m}{n + \kp(n)}\right)^{32(n - \eta)\log(n)}\]
Taking a union bound over sets of workers and sets of firms, the probability that there exists a set of workers of size $m$ such that exactly $\eta$ firms point to it, with $\eta < 2m$ is at most
\begin{eqnarray*}
    \binom{n + \kp(n)}{\eta} \binom{n}{n - \eta}\left(\frac{n + \kp(n) - m}{n + \kp(n)}\right)^{32(n - \eta)\log(n)} & = &\\
    \binom{n + \kp(n)}{\eta} \binom{n}{\eta}\left(1 - \frac{m}{n + \kp(n)}\right)^{32(n - \eta)\log(n)} & \le & \\
    \binom{2n}{\eta}^2 \left(1 - \frac{m}{2n}\right)^{32(n - \eta)\log(n)} & \le & \\
        \binom{2n}{2m}^2 \left(1 - \frac{1}{2n}\right)^{16 n \log(n)} & \le & \\
        \left(\frac{2ne}{2m}\right)^{4m} e^{-8  m \log(n)} & \le & \\
                (ne)^{4m} e^{-8  m \log(n)} & \le & \\
                (ne)^{4m} \left(\frac{1}{n}\right)^{8m} & \le & 1/n^3.
\end{eqnarray*}
Taking a union bound over the values of $m$ and $\eta$ finishes the proof.
\end{proof}

The following claim is the only place where we use the continuity of $g$, the density of the distribution $G$ used to sample $\alpha$.
\begin{claim} \label{loss}
With probability $1 - 1/n$, if firm $f$ points to a worker $w$ then $\alpha_{f,w} > 1 - \loss$.
\end{claim}

\begin{proof}
Let $\beta = 32\frac{n + \kp(n)}{n}\log(n)$, let $\thresh$ be such that $G(\thresh) = 1 - 8 \beta / (n + \kp(n)) = 1 - \frac{256 \log(n)}{n}$. The probability that out of $n + \kp(n)$ samples, exactly $\eta < \out$ of them would be above $\thresh$ is
\begin{eqnarray*}
\binom{n + \kp(n)}{\eta}\left(8 \beta/(n + \kp(n)\right)^\eta \left(1 - 8 \beta/(n + \kp(n)\right)^{(n + \kp(n) - \eta)} & \le & \\
\binom{n + \kp(n)}{\beta}\left(8 \beta/(n + \kp(n)\right)^\beta \left(1 - 8 \beta/(n + \kp(n)\right)^{\frac{n + \kp(n)}{2}} & \le & \\
\left(\frac{e(n + \kp(n))}{\beta}\right)^\beta \left(8 \beta/(n + \kp(n)\right)^\beta \left(1 - 8 \beta/(n + \kp(n)\right)^{\frac{n + \kp(n)}{2}} & = & \\
(8e)^\beta \left(1 - 8 \beta/(n + \kp(n)\right)^{\frac{n + \kp(n)}{2}} & \le & \\
(8e)^\beta e^{-4 \beta} & = & \\
e^{-(4 - \log 8)\beta} < 1/n^3.
\end{eqnarray*}
Taking a union bound over all possible values of $\eta < \beta$ and $n$ firms gives the result.

Since $g$ is continuous, there is an interval $[a,1]$ in which $g(x) \ge g(1) /2$. We choose $\thresh \ge 1 - \frac{512\log(n)}{g(1) n}$.
\end{proof}

We now use the graph and its properties to lower bound the profits of the firms.

\begin{claim} \label{graph_property}
If a firm $i$ points to worker that's unemployed then $u_i \ge 1 - \loss$. Moreover, if firm $i$ points to a worker $j$, and $j$ points to $i'$ then $u_i \ge u_i' - \loss$.
\end{claim}

\begin{proof}
If $i$ points to $j$ which is unemployed, $i$ could hire $j$ and pay her nothing. Therefore, it must be that $u_i \ge \alpha_{i,j} \ge 1 - \loss$.

For the other case, $i$ points to $j$, and $j$ points to $i'$. since $\alpha^n_{i',j} \le 1$, it must be that $v_j \le 1 - u_{i'}$. But since it is an equilibrium
\[u_i \ge \alpha^n_{i,j} - v_j \ge \alpha^n_{i,j} - 1 + u_{i'} \ge 1 - \loss - 1 + u_{i'} = u_{i'} - \loss,\]
where we used Claim~\ref{loss}
\end{proof}

Note that in particular if there is a path of length $k$ from a firm $i$ to a firm $i'$ then $u_i \ge u_i' - \cdot k $, and if there is a path of length $k$ from firm $i$ to an unemployed worker than $u_i \ge 1 - \loss k$.

\begin{proof}[Proof of \autoref{th:approx_lop}]

Consider any core outcome and two firms $i,i'$. We want to show that $u_i \ge u_{i'} - \frac{1024\log^2(n)}{g(1) n}$.

Let $S_1 = N(i)$ be the set of workers which $i$ points to. If there is $w \in S_1$ such that $\mu(w) = \emptyset$, then by Claim~\ref{graph_property} $u_i \ge \alpha_{i,w} \ge 1 - \loss$, and since $u_j \le 1$, we are done.

If for every $w \in S_1$ we have $\mu(w) \neq \emptyset$, then $\left|N(S_1)\right| = \left|S_1\right|$. In this case, let $S_2 = N\left(N(S_1)\right)$, and by Claim~\ref{expander} (or Claim~\ref{expander_large_kappa} if $\kp(n) > n$) we have $|S_2| \ge 2 \cdot |S_1|$. If there is $w \in S_2$ such that $\mu(w) = \emptyset$, there is a path of length 2 to an unemployed worker, and  $u_i \ge 1 - \frac{1024\log(n)}{g(1) n}$.

Continuing this by induction, if for some $w \in S_k$ we have $\mu(w) = \emptyset$ then $u_i \ge 1 - k \loss$. Otherwise, we define $S_{k+1} = N\left(N(S_k)\right)$, and get that $|S_{k+1}| \ge \nneigh |S_{k}|$. If $\kp(n) > n$, We continue with this process until we get to an unemployed worker and we are done in $\log(n)$ steps.

If $\kp(n) \le n$, we can continue for at most $log(n) - 3$ steps, after which we would be left with a set of firms $\FS$ of size $n/8$.

We now start with the target firm $i'$. There is a worker $w$ that points to it. According to Claim~\ref{point_to_workers}, there are at least two firms that point to $w$. Now there are two workers who point to those firms, four firms that point to them, etc. We continue this for $\log(n) - 4$ steps until we get to a set of workers $\WS$ of size at least $n/16$. However, invoking the second part of Claim~\ref{expander}, the set $\FS$ points to at least $0.99(n + \kp(n))$ workers, and since $|\WS| > 0.01 (n + \kp(n))$ by the pigeonhole principle there is an edge from $\FS$ to $\WS$, and a path of length at most $2 \log(n) - 6$ from $i$ to $i'$, which implies the theorem for firms.

We now show the second part of the theorem, namely that workers get similar salaries. If the market is balanced ($\kp(n) = 0$), we can apply the same argument. Else, let $w$ be a worker, and we show that $v_w \le \frac{1024\log^2(n)}{g(1) n}$. If $w$ is unemployed then $v_w = 0$ and we are done. Else, let $f$ be the firm which hires $w$.

If $\kp(n) > n$, the firm $f$ has a path of length $\log(n)$ to an unemployed worker and therefore $u_f \ge 1 - \frac{512\log^2(n)}{g(1) n}$ and therefore $v_w \le \frac{512 \log^2 n}{n}$.

If $\kp(n) \le n$, let $w'$ be an unemployed worker, which has a firm $f'$ pointing to it. Since for $\kp(n) \ge n$, any two firms have a path of length at most $2 \log(n) - 6$ between them we have

\[u_f \ge u_{f'} - (2\log(n) - 6) \loss,\]

and since $f'$ points to an unemployed worker $u_{f'} \ge 1 - \loss$. Therefore $u_f \ge 1 - (2\log(n) - 4) \loss$. But since $u_f + v_w \le 1$, it must be that $v_w \le \frac{1024\log^2(n)}{g(1) n}$ as required.
\end{proof}

\section{Applications} \label{sec:applications}

\subsection{Balanced markets}

We first deal with the case of $k(n) \equiv 0$. In this case we will show that whp the firm-optimal core allocation gives almost all the surplus to the firms. From symmetry this will imply that the core is ``long'' in the sense that different market mechanisms can lead to very different divisions of the surplus.

\begin{corollary} \label{cor:balanced_shares}
In a balanced market, with probability $1 - O\left(\frac{1}{n}\right)$, under the firm-optimal core allocation the workers' share of the surplus is $O\left(\frac{\log^2(n)}{n}\right)$.
\end{corollary}
\begin{proof}
Under the firm-optimal core allocation at least one worker must get exactly $0$ (or else we can substract some small $\epsilon > 0$ from each worker's salary and remain with a new core allocation which is even better for the firms). \autoref{th:approx_lop} then immediately implies that the sum of workers' salaries is $O(\frac{\log^2(n)}{n})$, whereas routine arguments from the theory of random assignment problems prove that the total surplus is $n - o(n)$.
\end{proof}

For the benchmark case of uniform distribution we can also provide tight bounds for the sum of workers' salaries under the firm-optimal core allocation.\footnote{In fact, the same proof strategy works, with minor modifications, for any distribution with continuous and strictly positive density. Furthermore, it also applies to unbalanced markets in which the number of firms is smaller than the number of workers.}

\begin{theorem} \label{th:balanced_fosm_bound}
In a balanced market with uniform distributions, $G = U[0,1]$, let $\psi^F\left(M^n\right) = \left(\mu^{F,n},u^{F,n},v^{F,n}\right)$ be the firm-optimal core allocation. Then
\begin{enumerate}
\item $E \left[ \sum_{i \in \{1,\dots,n\}} v^{F,n}_i \right] \leq \log(n)$, and
\item $E \left[ \sum_{i \in \{1,\dots,n\}} v^{F,n}_i \right] = \Omega\left(\log(n)\right)$.
\end{enumerate}
\end{theorem}
\begin{proof}
For the upper bound, consider a variant of the approximation algorithm suggested by \cite{ck1981} to solve a generalized version of the assignment game, in which firms are ordered from $1$ to $n$, and at each round only the lowest-number firm that still wants to propose actually proposes. What is the expected number of proposals between the first proposal of firm $k$ and the first proposal of firm $k+1$? Assume that it is now firm $f$'s turn to propose, and its previous aspiration level (i.e.\ the maximal utility it would get by giving some worker her current salary) was $u_i$. If for some unmatched $w \in W^n$ we have $\alpha^n_{f,w} \in [u-\epsilon,u)$ then firm $f$ will propose to that worker within at most $k$ rounds (or otherwise it will temporarily match to some other worker and be replaced by another proposing firm). Therefore the probability not proposing to some unmatched worker is above below by
\[
1 - \left(1 - \epsilon\right)^{n - k + 1} \approx (n-k+1)\epsilon,
\]
and so the expected number of steps until the first proposal of firm $k+1$ is at most $\frac{1}{(n-k+1)\epsilon}$, and the the expected raise in workers' salaries is $\frac{1}{n-k+1}$. Summing over $k$ we get that the expected sum of workers' salaries is at most
\[
1 + \frac{1}{2} + \dots + \frac{1}{n-1} + \frac{1}{n} \approx \log(n).
\]
For the lower bound a similar exercise can be taken, but taking into account that whp all workers' temporary salaries never exceed $1 - \frac{c\log^2(n)}{n}$. The details of the proof are omitted. 
\end{proof}

\subsection{Unbalanced markets} \label{subsec:applications_unbalanced}

The same reasoning that works for balanced markets can be applied to unbalanced markets. Our next result states that even slightly unbalanced markets are likely to give all the surplus to the short side.

\begin{corollary} \label{cor:unbalanced_shares}
In an unbalanced market, with probability $1 - O\left(\frac{1}{n}\right)$, under any core allocation all agents on the long side get $O\left(\frac{\log^2(n)}{n}\right)$, and the long side's share of the surplus is $O\left(\frac{\log^2(n)}{n}\right)$ as well.
\end{corollary}
\begin{proof}
Immediate from \autoref{th:approx_lop} and the fact that under any core allocation at least one agent on the long side remains unmatched.
\end{proof}

\section{Simulations} \label{sec:simulations}

While the bound established in \autoref{th:approx_lop} converges to zero quite rapidly, it is still worth while to see how powerful is the phenomenon described. Indeed, as we stress in the concluding section, we expect the wage dispersion to behave like $O\left(\frac{\log(n)}{n}\right)$. In this section we present results of computerized simulations that demonstrate how quick is the contraction of payoffs' dispersion, and how this affects the market. We then use more simulations to explore payoff dispersion in the presence of unbounded distributions. Each figure is based on averaging 400 trials for each data point.

We first focus on the benchmark case of uniform $[0,1]$ distribution, and study wage dispersion in balanced markets under the firm-optimal core allocation. The right panel of \autoref{fig:balanced_uniform} shows that in this case both the mean salary of workers and the maximum salary any worker gets (which is also the maximal difference between any two workers' salaries) approach zero very quickly. The left panel of the same figure exemplifies the fact that the core in balanced markets is long, as suggested by \autoref{cor:balanced_shares}.

\begin{figure}[htp]
\centering
\includegraphics[width=60mm]{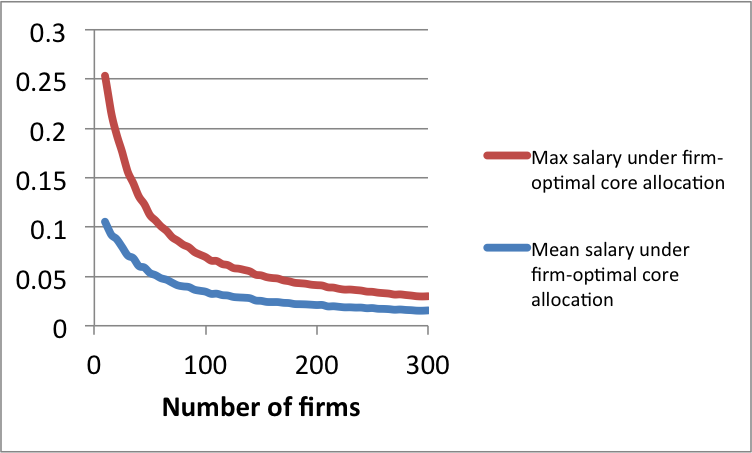}
\ \ \
\includegraphics[width=60mm]{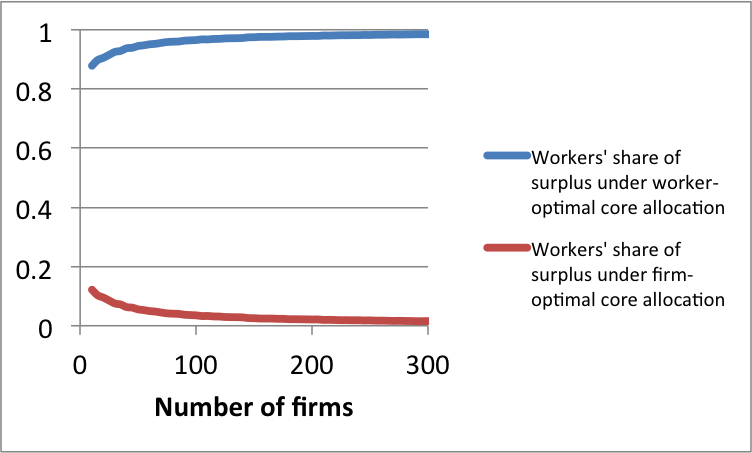}
\caption{Balanced markets ($n$ firms, $n$ workers), distribution: $U[0,1]$}
\label{fig:balanced_uniform}
\end{figure}

In unbalanced markets we expect the core to be much more narrow, per \autoref{cor:unbalanced_shares}. The left panel of \autoref{fig:unbalanced_uniform} shows that when the number of workers is only slightly larger than the number of firms, both the workers' mean salary and the maximum salary any worker gets, approach zero rapidly, even under the worker-optimal core allocation. Furthermore, as the right panel demonstrates, in this case the workers' share in the surplus approaches 0, even under the worker-optimal core allocation. Finally, \autoref{fig:unbalanced_uniform_shift} parallels Figure~4 of \cite{akl2013}, and depicts workers' share of the surplus when the number of workers is constant at $50$, and the number of firms varies from $20$ to $80$.

\begin{figure}[htp]
\centering
\includegraphics[width=60mm]{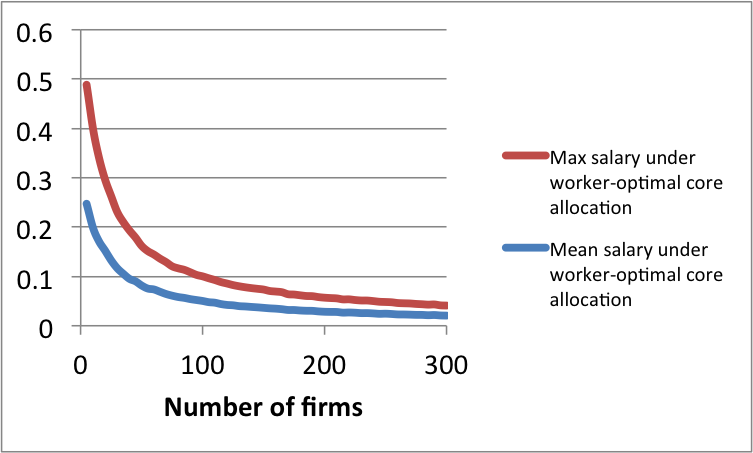}
\ \ \
\includegraphics[width=60mm]{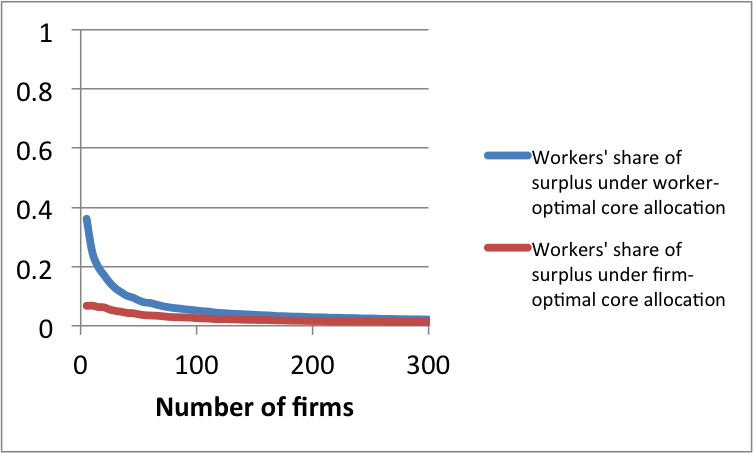}
\caption{Unbalanced markets ($n$ firms, $n+1$ workers), distribution: $U[0,1]$}
\label{fig:unbalanced_uniform}
\end{figure}

\begin{figure}[htp]
\centering
\includegraphics[width=60mm]{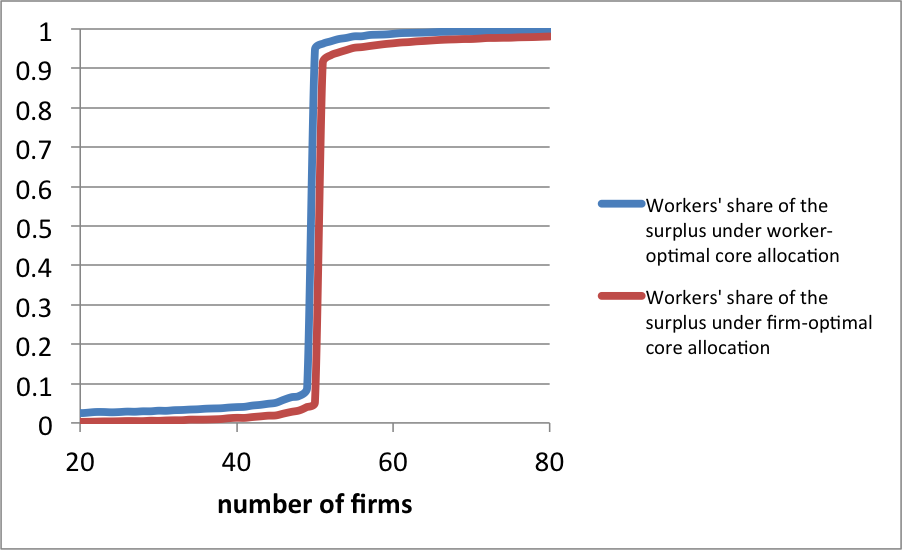}
\caption{Unbalanced markets ($50$ workers), distribution: $U[0,1]$}
\label{fig:unbalanced_uniform_shift}
\end{figure}

\subsection{Unbounded distributions}

A very interesting question is what happens to our results when the distribution $G$ is unbounded. A rough intuition for this case is that unbounded distributions with a heavy tail may create two types of outliers in the market. ``Good'' outliers are agents that are highly productive compared to others, such that agents from the other side fiercely compete on being matched with them. These agents share a significant portion of the surplus they help create, and if they are common enough, they may offset other forces that would otherwise squeeze the surplus from their side (such as an adversarial core allocation, or slight imbalance in favor of the other side of the market). ``Bad'' outliers are agents that are so unproductive in comparison with their peers, that being matched with them becomes the effective reservation value for agents on the other side. If for any reason there are many such agents on one of the sides, then the residual market behaves as if it was unbalanced, and even mildly productive agents that are on the same side as the less-productive agents earn a sizable part of their contribution to the total surplus.

\autoref{fig:unbounded_salaries} describes workers' salaries in a balanced market governed by two distributions: Exponential and Weibull with parameter $0.25$ (a heavy-tailed distribution). As expected, we can see that the maximal salary does not go to $0$, and behaves more like $\log(n)$. In \autoref{fig:unbounded_shares} we can see the workers' share in the total surplus for the same markets.

\begin{figure}[htp]
\centering
\includegraphics[width=60mm]{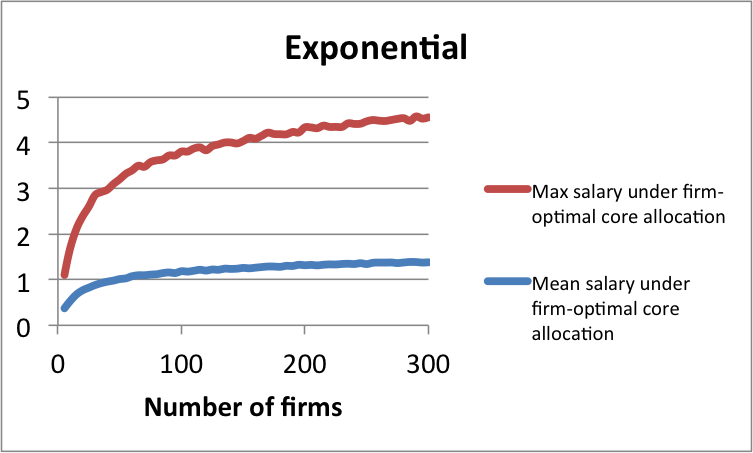}
\ \ \
\includegraphics[width=60mm]{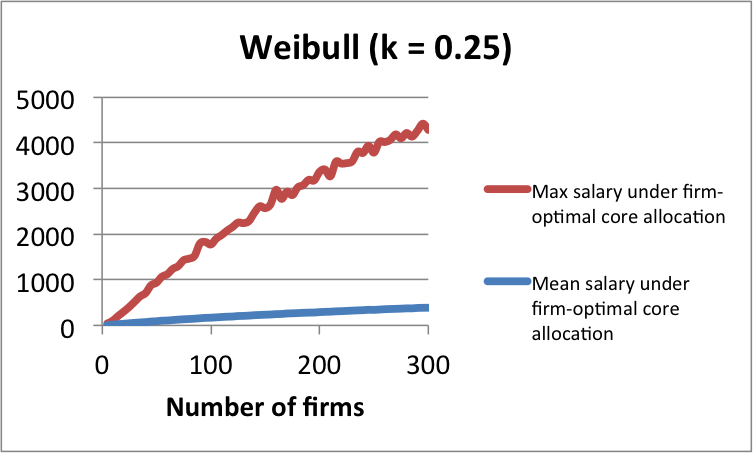}
\caption{Balanced markets, workers' salaries}
\label{fig:unbounded_salaries}
\end{figure}

\begin{figure}[htp]
\centering
\includegraphics[width=60mm]{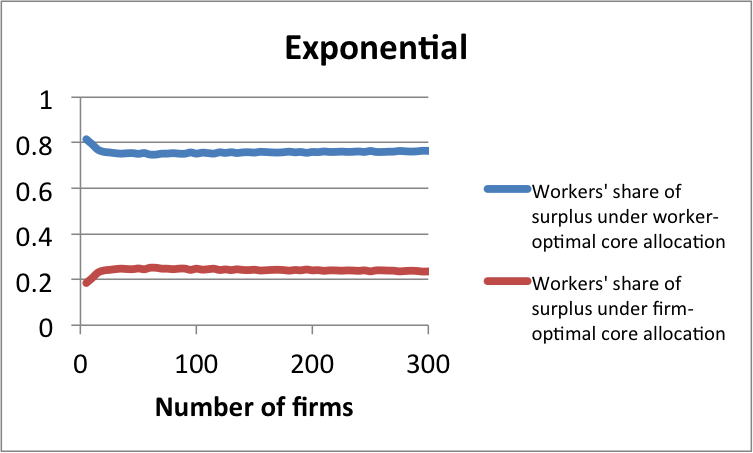}
\ \ \
\includegraphics[width=60mm]{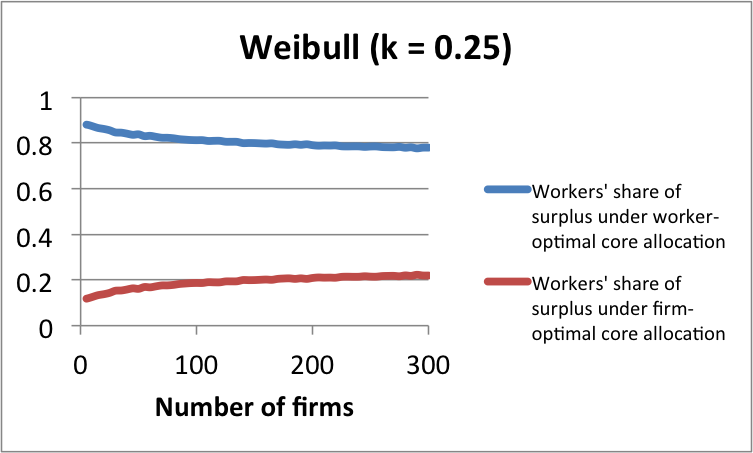}
\caption{Balanced markets, workers' share of surplus}
\label{fig:unbounded_shares}
\end{figure}

We conclude this section by noting that it is relatively easy to establish analytically that \autoref{th:approx_lop} does not hold for the case of an exponential distribution:

\begin{proposition} \label{prp:exp_good_outliers}
If $G$ is the exponential distribution and the market is balanced and governed by the firm-optimal core allocation, then whp there are two workers $i$ and $j$ such that $\left|v_i - v_j\right| = \Omega\left(\log(n)\right)$.
\end{proposition}
\begin{proof}
Let $p^n$ denote the probability that for a given firm $f \in F^n$ there exists a worker $w \in W^n$ such that $\alpha^n_{f,w} > 1.1\log(n)$ and $\max_{w' \in W^n \setminus \{w\}} \alpha^n_{f,w'} < \log(n)$. Then
\[
p = n \cdot e^{-1.1\log(n)} \cdot \left(1 - e^{-\log(n)}\right)^{n-1} = \frac{1}{n^{0.1}} \cdot \left(1-\frac{1}{n}\right)^{n-1} \approx \frac{1}{e n^{0.1}}.
\]
This specifically implies that whp there are $\Omega\left(n^{0.8-\epsilon}\right)$ firms that meet the above condition. If any of those two firms happen to have the same worker being the outlier, then this worker must get paid at least $0.1\log(n)$ under any core allocation. Since there are $\Omega\left(n^{1.6-2\epsilon}\right)$ pairs, then we get that there are many workers that get paid $\Theta\left(\log(n)\right)$. Finally, at least one worker's salary is $0$ under the firm-optimal core allocation, and so we are done.
\end{proof}

\subsection{The role of outliers}

To give some intuition to the way wage dispersion looks, and to the role of ``Good'' and ``Bad'' outliers, \autoref{fig:histogram_salaries} presents the wage distribution of workers in a single instance of a firm-optimal balanced markets ($n = 1000$). In the left pane, $G$ is the uniform distribution, while in the right pane $G$ is the exponential distribution.

\begin{figure}[htp]
\centering
\includegraphics[width=60mm]{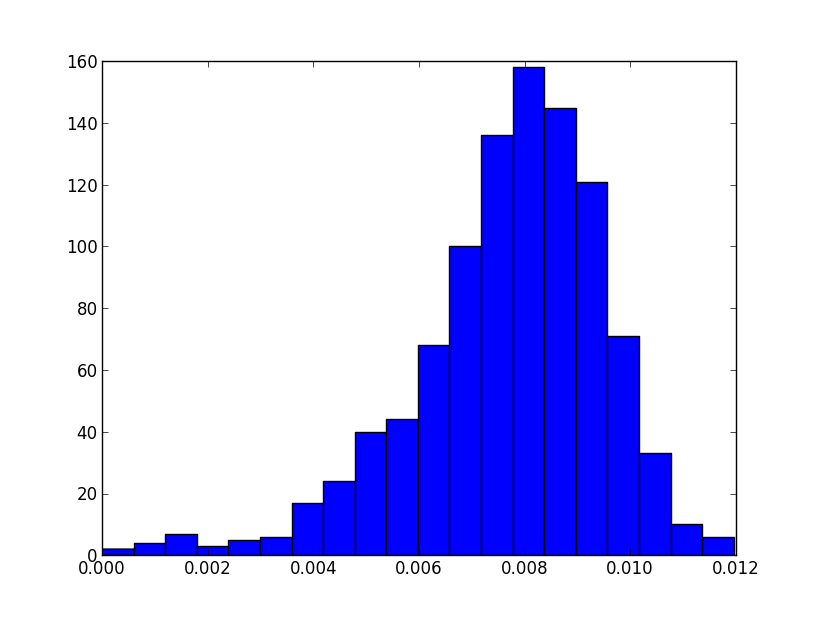}
\ \ \
\includegraphics[width=60mm]{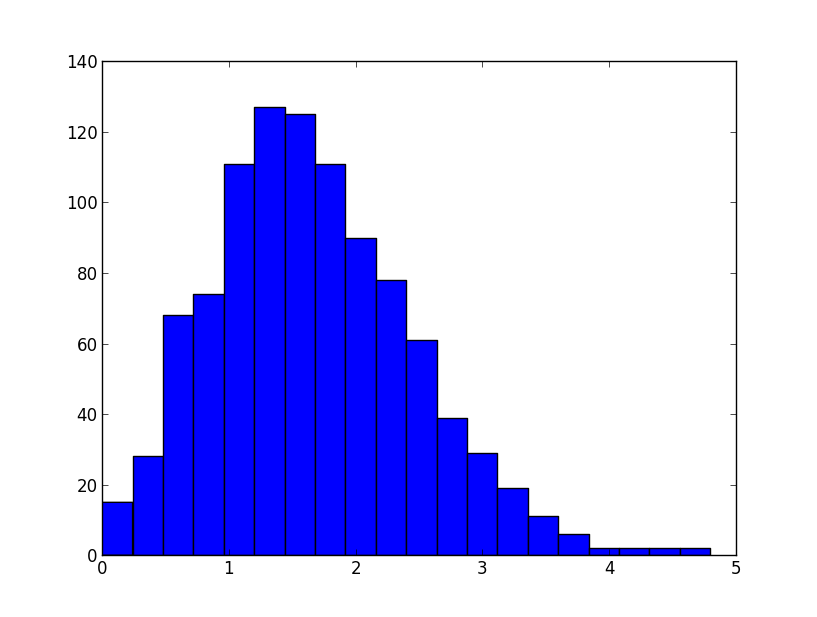}
\caption{Balanced markets, workers' salaries}
\label{fig:histogram_salaries}
\end{figure}

It is not surprising that salaries are higher in the exponential distribution, and that the total surplus is higher. Two other conclusions one can draw are:

\begin{enumerate}
\item The wage distribution in the uniform distribution is balanced, where in the exponential distribution it has a tail to the right
  \item Although the wages in the exponential distribution have a tail, it is not heavy, and most of the salary comes from the majority of the population (there is no $1\%$ which takes the economic pie). We see this as evidence that the contribution to the workers' salaries come from the existence of ``Bad'' outliers and their effect on the population.
\end{enumerate}



In \autoref{fig:histogram_salaries}  it is hard to see the effect that ``Bad'' outliers have in the uniform distribution. To better understand the effect, \autoref{fig:worst_players} depicts the relationship between the sum of salaries of the workers in a firm-optimal outcome, and the quality of the ``Worst'' workers. The left pane shows the the dependency of the sum of salaries on the quality of the ``Worst'' worker (minimum over workers of the maximum value that worker has with a firm), and is based on $100$ trials with $n = 100$. In the right pane, we present the effect of the second ``Worst'' worker. To this end, we sampled $100$ random markets, for which the quality of the worst worker was $0.95 \pm 0.00001$, and depicted the dependency of the sum of salaries in the second ``Worst'' worker.\footnote{We did rejection sampling, that is we sampled random markets, and kept them if the quality of the worst worker was between $0.94999$ and $0.95001$. This required us around a million samples.}

\begin{figure}[htp]
\centering
\includegraphics[width=60mm]{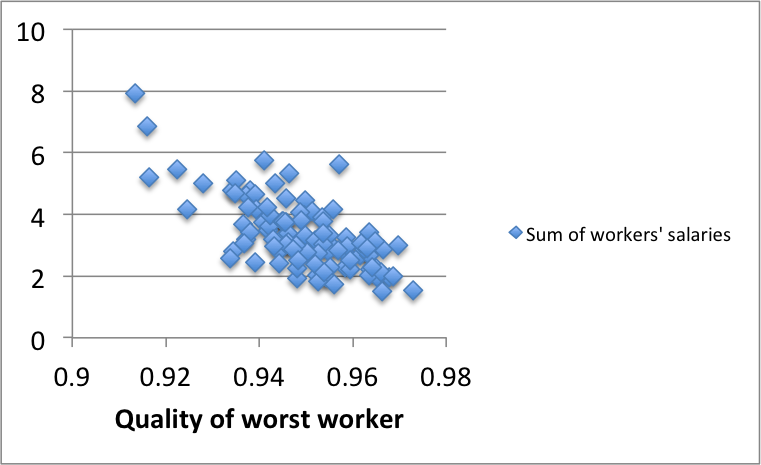}
\ \ \
\includegraphics[width=60mm]{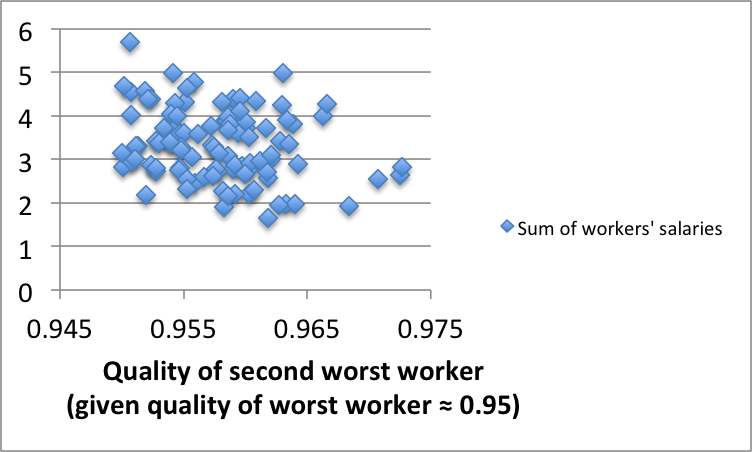}
\caption{Balanced markets ($n$ firms, $n$ workers), distribution: $U[0,1]$}
\label{fig:worst_players}
\end{figure}

Two observations about the graphs are

\begin{enumerate}
  \item The dependency between the total salary and the quality of the worst worker becomes linear (with a slope of $n$) as the quality of the worker becomes lower. The reason is that it serves as a ``reservation value'' for the market, and reducing its quality by $\epsilon$ increases (almost) all wages by (almost) $\epsilon$. We note that when the quality of the worst worker is high, sometimes other workers serve as reservation value, and therefore the slope is not linear.
  \item The second ``Worst'' worker also has a linear effect on the market, although here it is not as prominent.
\end{enumerate}

We note that fitting the first graph to a linear approximation gives an $R^2$ value of $0.48$, and an $R^2$ value of 0.07 for the second graph. The second highest value of the ``Worst'' worker also has an effect on the market. We conjecture that the sum of salaries that the workers get is greatly affected by the $\log(n)$ ``Worst'' workers.

We also note that the quality of the ``Worst'' firm (minimum over the firms of the maximal value this firm has with a worker) has little effect on the sum of salaries the workers get.

\section{Conclusion} \label{sec:conclusion}

As hinted in the introduction, there is an abundance of real-world data suggesting that not only does wage inequality exist, but that it persists across nations, industries and even within firms (see, e.g.\ \cite{mortensen2005} and references therein). While investment and variation in human capital explain a significant portion of the differences in income, there are residual discrepancies that are not so easily explained. Economists have considered many possible approaches to this issue (for example, compensating differentials, motivating effort, complementarities, search frictions, and more) and usually carried an analysis under extremely restrictive assumptions on match value heterogeneity. Our contribution provides tools that may allow interacting those explanations within a more realistic framework.

To give but two examples, our model inherently allows for compensating differentials (by changing reservations values and then applying the analysis), and it is relatively straightforward to add human capital or quality of firms, as long as the production function is separable in those two factors. If, say, workers' quality is captured by a linear (common) contribution to each element in the relevant column, it is easy to show that each worker will get its contribution above the lowest type, and the rest of the surplus will be distributed according to similar principles to those described in this paper. Non-separable production function (that is, complementarities) require more work because the proof must take into account that assortative matching will take place to some extent (see also below).

There are several open questions that present future challenges for our model. The first is providing a tighter bound on wage dispersion. The proof of \autoref{th:approx_lop} used the upper $\frac{\log(n)}{n}$ fraction in order to create an expander, and used it to find paths between different firms. One could try to do so based on the upper $\frac{c}{n}$ fraction, and try to get the correct bound of $\frac{\log(n)}{n}$. There are a couple of obstacles to this approach. The first, is that for some firms using the top $\frac{c}{n}$ edges would already create a loss of $\frac{\log(n)}{n}$ (see Claim~\ref{loss}). A potentially bigger obstacle is that the edges from workers to firms are not random, and are chosen by the allocation. Our proof treats them as adversarial, but this requires us to use more edges from firms to workers (see Claim~\ref{point_to_workers}).

A different direction is studying the induced salary distribution for the exponential distribution and for general unbounded distributions. As mentioned in \autoref{sec:simulations}, unbounded distributions call for new intuitions that may offset the current analysis, and studying them is interesting. Specifically for the exponential distribution, we believe that the sum of workers' salaries in a balanced market governed by the firm-optimal core allocation is $\Theta\left(n\log\left(\log(n)\right)\right)$. The (rough) intuition is that the $\min\max$ over the productivity matrix' entries behaves like $\log(n) - \log\left(\log(n)\right)$, whereas solving the assignment problem would give a mean surplus of $\log(n)$. Moreover, in expectation only $\frac{n}{\log(n)}$ of the agents would have a value above $\log(n) + \log(\log(n))$, and since the maximal value in the matrix is $2 \log(n)$ their total salary can not exceed $2n$.\footnote{The procedure used in the proof of \autoref{prp:exp_good_outliers} only used the ``good'' outliers, and we think that most of the salary comes due to the existence of bad outliers - see \autoref{fig:histogram_salaries}.}

Finally, it is also interesting to analyze correlated random matrices. As mentioned above, adding correlation can significantly distort some of our results. Adding separable noise with agent-specific mean can easily be dealt with, however interaction variables (such as Cobb-Douglas production function) are more tricky. With sufficient agents, we can get near-efficiency on both dimensions (the common element and the pair-specific productivity), but determining who gets what is more involved.

\bibliographystyle{authordate1}
\bibliography{MatchingBibliography}

\begin{thebibliography}{}

\bibitem[\protect\citename{Aldous, }2001]{aldous2001}
Aldous, David~J. 2001.
\newblock The $\zeta$ (2) limit in the random assignment problem.
\newblock {\em Random Structures \& Algorithms}, {\bf 18}(4), 381--418.

\bibitem[\protect\citename{Ashlagi {\em et~al.}, }2011]{abh2011}
Ashlagi, Itai, Braverman, Mark, and Hassidim, Avinatan. 2011.
\newblock Matching with couples revisited.
\newblock {\em pp.  335--336 in} {\em Proceedings of the 12th ACM conference on
  Electronic commerce}.
\newblock ACM.

\bibitem[\protect\citename{Ashlagi {\em et~al.}, }2013]{akl2013}
Ashlagi, Itai, Kanoria, Yashodhan, and Leshno, Jacob~D. 2013.
\newblock Unbalanced random matching markets.
\newblock {\em pp.  27--28 in} {\em ACM Conference on Electronic Commerce}.

\bibitem[\protect\citename{Che and Kojima, }2010]{ck2010}
Che, Yeon-Koo, and Kojima, Fuhito. 2010.
\newblock Asymptotic equivalence of probabilistic serial and random priority
  mechanisms.
\newblock {\em Econometrica}, {\bf 78}(5), 1625--1672.

\bibitem[\protect\citename{Chen {\em et~al.}, }2012]{cfy2012}
Chen, Bo, Fujishige, Satoru, and Yang, Zaifu. 2012.
\newblock {\em Decentralized market processes to stable job matchings with
  competitive salaries}.
\newblock mimeo.

\bibitem[\protect\citename{Coles and Shorrer, }2014]{cs2014}
Coles, Peter, and Shorrer, Ran. 2014.
\newblock Optimal truncation in matching markets.
\newblock {\em Games and Economic Behavior}.
\newblock forthcoming.

\bibitem[\protect\citename{Coles {\em et~al.}, }2014]{cgs2014}
Coles, Peter, Gonczarowski, Yannai, and Shorrer, Ran. 2014.
\newblock {\em Strategic Behavior in Unbalanced Matching Markets}.
\newblock mimeo.

\bibitem[\protect\citename{Coppersmith and Sorkin, }1999]{cs1999}
Coppersmith, Don, and Sorkin, Gregory~B. 1999.
\newblock Constructive bounds and exact expectations for the random assignment
  problem.
\newblock {\em Random Structures \& Algorithms}, {\bf 15}(2), 113--144.

\bibitem[\protect\citename{Crawford and Knoer, }1981]{ck1981}
Crawford, Vincent~P, and Knoer, Elsie~Marie. 1981.
\newblock Job matching with heterogeneous firms and workers.
\newblock {\em Econometrica: Journal of the Econometric Society},  437--450.

\bibitem[\protect\citename{Demange and Gale, }1985]{dg1985}
Demange, Gabrielle, and Gale, David. 1985.
\newblock The Strategy Structure of Two-Sided Matching Markets.
\newblock {\em Econometrica}, {\bf 53}(4), 873--888.

\bibitem[\protect\citename{Demange {\em et~al.}, }1986]{dgs1986}
Demange, Gabrielle, Gale, David, and Sotomayor, Marilda. 1986.
\newblock Multi-item auctions.
\newblock {\em The Journal of Political Economy},  863--872.

\bibitem[\protect\citename{Gale and Shapley, }1962]{gs1962}
Gale, David, and Shapley, Lloyd~S. 1962.
\newblock College Admissions and the Stability of Marriage.
\newblock {\em The American Mathematical Monthly}, {\bf 69}(1), 9--15.

\bibitem[\protect\citename{Immorlica and Mahdian, }2005]{im2005}
Immorlica, Nicole, and Mahdian, Mohammad. 2005.
\newblock Marriage, honesty, and stability.
\newblock {\em pp.  53--62 in} {\em Proceedings of the sixteenth annual
  ACM-SIAM symposium on Discrete algorithms}.
\newblock Society for Industrial and Applied Mathematics.

\bibitem[\protect\citename{Karp, }1987]{karp1987}
Karp, Richard~M. 1987.
\newblock An upper bound on the expected cost of an optimal assignment.
\newblock {\em pp.  1--4 in} {\em Discrete algorithms and complexity:
  Proceedings of the Japan-US Joint Seminar},  vol. 15.
\newblock Academic Press, New York.

\bibitem[\protect\citename{Klaus and Payot, }2012]{kp2012}
Klaus, Bettina, and Payot, Fr{\'e}d{\'e}ric. 2012.
\newblock {\em Paths to stability in the assignment problem}.
\newblock mimeo.

\bibitem[\protect\citename{Knuth, }1976]{knuth1976}
Knuth, Donald~Ervin. 1976.
\newblock {\em Mariages stables et leurs relations avec d\&autres probl{\`e}mes
  combinatoires}.
\newblock Presses de l'Universit{\'e} de Montr{\'e}al.

\bibitem[\protect\citename{Kojima and Manea, }2010]{km2010}
Kojima, Fuhito, and Manea, Mihai. 2010.
\newblock Incentives in the probabilistic serial mechanism.
\newblock {\em Journal of Economic Theory}, {\bf 145}(1), 106--123.

\bibitem[\protect\citename{Kojima and Pathak, }2009]{kp2009}
Kojima, Fuhito, and Pathak, Parag~A. 2009.
\newblock Incentives and Stability in Large Two-Sided Matching Markets.
\newblock {\em American Economic Review}, {\bf 99}(3), 608--27.

\bibitem[\protect\citename{Kojima {\em et~al.}, }2013]{kpr2013}
Kojima, Fuhito, Pathak, Parag~A., and Roth, Alvin~E. 2013.
\newblock Matching with Couples: Stability and Incentives in Large Markets.
\newblock {\em The Quarterly Journal of Economics}, {\bf 128}(4), 1585--1632.

\bibitem[\protect\citename{Krokhmal and Pardalos, }2009]{kp2009assignment}
Krokhmal, Pavlo~A, and Pardalos, Panos~M. 2009.
\newblock Random assignment problems.
\newblock {\em European Journal of Operational Research}, {\bf 194}(1), 1--17.

\bibitem[\protect\citename{Lee, }2014]{lee2014}
Lee, SangMok. 2014.
\newblock {\em Incentive compatibility of large centralized matching markets}.
\newblock mimeo.

\bibitem[\protect\citename{Lee and Yariv, }2014]{ly2014}
Lee, SangMok, and Yariv, Leeat. 2014.
\newblock {\em On the Efficiency of Stable Matchings in Large Markets}.
\newblock mimeo.

\bibitem[\protect\citename{Manea, }2009]{manea2009}
Manea, Mihai. 2009.
\newblock Asymptotic ordinal inefficiency of random serial dictatorship.
\newblock {\em Theoretical Economics}, {\bf 4}(2), 165--197.

\bibitem[\protect\citename{Mo, }1988]{mo1988}
Mo, Jie{\-}Ping. 1988.
\newblock Entry and Structures of Interest Groups in Assignment Games.
\newblock {\em Journal of Economic Theory}, {\bf 46}(1), 66 -- 96.

\bibitem[\protect\citename{Mortensen, }2005]{mortensen2005}
Mortensen, Dale. 2005.
\newblock {\em Wage dispersion: why are similar workers paid differently}.
\newblock The MIT Press.

\bibitem[\protect\citename{Nax {\em et~al.}, }2013]{npy2013}
Nax, Heinrich~H, Pradelski, Bary~SR, and Young, H~Peyton. 2013.
\newblock {\em Decentralized dynamics to optimal and stable states in the
  assignment game}.
\newblock mimeo.

\bibitem[\protect\citename{N{\'u}{\~n}ez and Rafels, }2008]{nr2008}
N{\'u}{\~n}ez, Marina, and Rafels, Carles. 2008.
\newblock On the dimension of the core of the assignment game.
\newblock {\em Games and Economic Behavior}, {\bf 64}(1), 290--302.

\bibitem[\protect\citename{Pittel, }1989]{pittel1989}
Pittel, Boris. 1989.
\newblock The average number of stable matchings.
\newblock {\em SIAM Journal on Discrete Mathematics}, {\bf 2}(4), 530--549.

\bibitem[\protect\citename{Pittel, }1992]{pittel1992}
Pittel, Boris. 1992.
\newblock On likely solutions of a stable marriage problem.
\newblock {\em The Annals of Applied Probability},  358--401.

\bibitem[\protect\citename{Quint, }1987]{quint1987}
Quint, Thomas. 1987.
\newblock {\em Elongation of the Core in an Assignment Game}.
\newblock Tech. rept. DTIC Document.

\bibitem[\protect\citename{Schwarz and Yenmez, }2011]{sy2011}
Schwarz, Michael, and Yenmez, M~Bumin. 2011.
\newblock Median stable matching for markets with wages.
\newblock {\em Journal of Economic Theory}, {\bf 146}(2), 619--637.

\bibitem[\protect\citename{Shapley, }1955]{shapley1955}
Shapley, Lloyd~S. 1955.
\newblock {\em Markets as cooperative games}.
\newblock Rand Corporation.

\bibitem[\protect\citename{Shapley and Shubik, }1971]{ss1971}
Shapley, Lloyd~S, and Shubik, Martin. 1971.
\newblock The assignment game I: The core.
\newblock {\em International Journal of Game Theory}, {\bf 1}(1), 111--130.

\bibitem[\protect\citename{Walkup, }1979]{walkup1979}
Walkup, David~W. 1979.
\newblock On the expected value of a random assignment problem.
\newblock {\em SIAM Journal on Computing}, {\bf 8}(3), 440--442.

\bibitem[\protect\citename{Wilson, }1972]{wilson1972}
Wilson, LB. 1972.
\newblock An analysis of the stable marriage assignment algorithm.
\newblock {\em BIT Numerical Mathematics}, {\bf 12}(4), 569--575.

\end{thebibliography}

\end{document}